\documentclass{endm}
\usepackage{endmmacro}
\usepackage{graphicx}

\newtheorem{observation}{Observation}

\begin{document}

\begin{verbatim}\end{verbatim}

\begin{frontmatter}

\title{Exact algorithms for dominating induced matchings}

\author{Min Chih Lin$^{3,1}$\thanksref{thxoscar}\thanksref{oscaremail} \qquad\qquad Michel J Mizrahi}
\address{CONICET, Instituto de C\'alculo and Departamento de Computaci\'on\\ Universidad de Buenos Aires\\ Buenos Aires, Argentina}


\author{Jayme L Szwarcfiter\thanksref{thxjayme}\thanksref{jaymeemail}}
\address{Inst de Matem\'atica, COPPE and NCE  \\ Universidade Federal do Rio de Janeiro \\
Rio de Janeiro, Brazil}

\thanks[oscaremail]{Email:
   \href{mailto:oscarlin@dc.uba.ar} {\texttt{\normalshape
   \{oscarlin,mmizrahi\}@dc.uba.ar}}}
\thanks[jaymeemail]{Email:
   \href{mailto:jayme@nce.ufrj.br} {\texttt{\normalshape
   jayme@nce.ufrj.br}}}
   
\thanks[thxoscar]{Partially supported by UBACyT Grants 20020100100754 and 20020090100149, PICT ANPCyT
Grant 1970 and PIP CONICET Grant 11220100100310.}
\thanks[thxjayme]{Partially supported by CNPq, CAPES and FAPERJ, brazilian research agencies. Presently visiting the Instituto Nacional de Metrologia, Qualidade e Tecnologia", Brazil.} 
  



\begin{abstract}
Say that an edge of a graph $G$ dominates itself and every other
edge adjacent to it. An edge dominating set of a graph $G=(V,E)$
is a subset of edges $E' \subseteq E$ which dominates all edges of
$G$. In particular, if every edge of $G$ is dominated by exactly
one edge of $E'$ then $E'$ is a dominating induced matching. It is
known that not every graph admits a dominating induced matching,
while the problem to decide if it does admit is NP-complete. In
this paper we consider the problem of finding a minimum weighted
dominating induced matching, if any,  of a graph with weighted
edges. We describe two exact algorithms for general graphs. The
algorithms are efficient in the cases where $G$ admits a known
vertex dominating set of small size, or when $G$ contains a
polynomial number of maximal independent sets.
\end{abstract}


\end{frontmatter}

\section{Introduction}\label{intro}
By $G(V,E)$ we denote a \emph{simple undirected graph} with vertex
set $V$ and edge set $E$, $n=|V|$ and $m=|E|$. We consider $G$ as a \emph{weighted}
graph, that is, one in which there is a non-negative real weight
assigned to each edge of $G$.
If $v \in V$ and $W \subseteq V$, then denote by $N(v)$, the set
of vertices adjacent (neighbors) to $v$, denote by $G[W]$ the
subgraph of $G$ induced by $W$, and write $N_W(v)=N(v) \cap W$.
Say that $D\subseteq V$ is a \emph{(vertex) dominating set} of $G$
if $D \cup N(D) = V$, where $N(D)=\bigcup _{v\in D} N(v)$.

Given an edge $e \in E$, say that $e$ \emph{dominates} itself and
every edge sharing a vertex with $e$. Subset $E'\subseteq E$ is an
\emph{induced matching} of $G$ if each edge of $G$ is dominated by
at most one edge in $E'$. A \emph{dominating induced matching
(DIM)} of $G$ is a subset of edges which is both dominating and an
induced matching. Not every graph admits a DIM, and the problem of
determining whether a graph admits it is also known in  the
literature as \emph{efficient edge  domination problem}. The
weighted version of DIM problem is to find a DIM such that the sum
of weights of its edges is minimum among all DIM's, if any.

The (unweighted version) of the dominating induced matching
problem is known to be NP-complete \cite{Gr-Sl-Sh-Ho}, even for
planar bipartite graphs \cite{Lu-Ko-Ta} or regular graphs
\cite{Ca-Ce-De-Si}. There are polynomial time algorithms for some
classes, as chordal graphs \cite{Lu-Ko-Ta}, generalized
series-parallel graphs \cite{Lu-Ko-Ta} (both for the weighted
problem), claw-free graphs \cite{Ca-Ko-Lo}, graphs with bounded
clique-width \cite{Ca-Ko-Lo}, convex graphs \cite{Ko}, bipartite
permutation graphs \cite{Lu-Ta} (see also \cite{Br-Lo}).

In this paper, we describe two exact (exponential time) algorithms
for the weighted problem. The first runs in linear time for a
given vertex dominating set of fixed size of the graph. The second
runs in polynomial time if the graph admits a polynomial number of
maximal independent sets.

We will use an alternative definition, taken from  \cite{Do-Lo},  of
the problem of finding a dominating induced matching. It asks to
determine if the vertex set of a graph $G$ admits a partition into
two subsets $W$ and $B$ such that $W$ is an independent set and
$B$ induces an 1-regular graph. The vertices of $W$ are called
\emph{white} and those  of $B$ are \emph{black}.

We consider only graphs without isolated vertices, because
isolated vertices must be white in any black-white partition. So,
we can delete them and solve the problem in the residual graph.

Assigning one of the two possible colors to vertices of $G$ is
called a coloring of $G$. A coloring is \emph{partial} if only
part of the vertices of G have been assigned colors, otherwise it
is \emph{total}. A partial coloring is {\it valid} if no two white
vertices are adjacent and no black vertex has more than one black
neighbor. A black vertex is \emph{single} if it has no black
neighbors, otherwise, it is \emph{paired}. A total coloring is
valid if no two white vertices are adjacent and every black vertex
is paired. Clearly, $G$ admits a DIM if and only if it admits a
total valid coloring. In fact, a total valid coloring defines
exactly one DIM, given by the set $B$.

For a coloring $C$ of the vertices of $G$, denote by
$C^{-1}(white)$ and $C^{-1}(black)$, the subsets of vertices
colored white and black. A coloring $C'$ is an {\it extension} of
a $C$ if  $C^{-1}(black)\subseteq C'^{-1}(black)$ and
$C^{-1}(white) \subseteq C'^{-1}(white)$.

\section{An Algorithm Based on Vertex Domination}

Next, we will propose an exact algorithm for solving the
weighted dominating induced matching problem, for general graphs,
based on vertex dominations.


Let $C$ be a \emph{partial} valid coloring of $G = (V,E)$. Such as
in \cite{Do-Lo}, this coloring can be further propagated according to
the following rules:

\begin{enumerate}
    \item each neighbor of a white vertex must be black
    \item Except for its pair, the neighbors of a paired black vertex must be
    white
    \item each vertex with two black neighbors must be white
    \item if a single black vertex has exactly one uncolored neighbor then this neighbor must be black
\end{enumerate}

We can  propagate the coloring iteratively, until it becomes no
more possible to color new vertices, simply by the application of
the above rules. Then we check for validity. A valid coloring so
obtained is then called {\it stable}.

Let $C$ be a partial valid coloring of $G$, and $C'$ be a stable
coloring obtained from $C$, by the applications of rules (i)-(iv),
above. Denote by $D$ and $D'$, respectively the subsets of vertices
of $G$ which are colored in $C$ and $C'$. Clearly, $D' \supseteq
D$. For our purposes, assume that the initial set $D$ of colored
vertices  is a vertex dominating set of $G$.

\begin{lemma} Let $C'$ be a stable coloring. Then
\begin{enumerate}
\item If there are no single (black) vertices then $C'$ is a total
coloring,  \item  Any uncolored vertex has exactly one black
neighbor, and such a neighbor must be single.
\end{enumerate}
\end{lemma}
\begin{proof}
   Recall that $D$ is the initial colored vertices and is a dominating set of the graph $G$. $C'$ is not a total coloring if only if there is some uncolored vertex $v$. Clearly, $v \not \in D$ and $N(v)\cap D \ne \emptyset$. Let $w$ be some neighbor of $v$ in $D$. If the color of $w$ is white then $v$ must be colored black by rule (i) which is a contradiction. Hence $w$ is a black vertex. Again, if $w$ is a paired black vertex or $v$ has another black neighbor $w'\ne w$, $v$ must get color white by rules (ii) and (iii) and this is a contradiction. Consequently, $v$ has exactly one black neighbor and which is single black vertex. Therefore, if there are no single black vertices then there are not uncolored vertices and $C'$ is a total coloring.
\end{proof}

Let $D'$ be the colored vertices of the stable coloring $C'$, let
$S = \{s_1, \ldots, s_p\}$ be the set of single vertices, and $U$
the set of still uncolored vertices of $G$, that is, $U = V
\setminus D'$. The above lemma implies that $U$ admits a partition
into (disjoint) parts:

\centerline{$ U = (N(s_1) \cap U) \cup \ldots \cup (N(s_p) \cap
U)$}


\begin{theorem}
Let $C$ be a coloring of the vertices of $G$, $C'$ a stable
extension of it, and $D = C^{-1}(black) \cup C^{-1}(white)$ a
dominating set of $G$. Then (i) $S \subseteq C^{-1}(black)$; and
(ii) if $C$ extends to a valid total coloring $C''$ then $C''$ is
an extension of $C'$.
\end{theorem}
\begin{proof}
Suppose that $S  \not \subseteq C^{-1}(black)$ which means that
exists a vertex $s_i \in S$ and $s_i \notin C^{-1}(black)$. By
definition of $S$, $s_i$ is a single black vertex. If $s_i \in D$
then $s_i \in C^{-1}(black)$, contradiction. Therefore $s_i \notin
D$.

Since $D$ is a dominating set, then $\exists v \in D$ such that
$s_i \in N(v)$. If $v$ is black then $s_i$ is not a single black
vertex, again a contradiction. Hence $v$ must be white.
\begin{itemize}
    \item If $s_i$ has no uncolored neighbors then $C$ is not extensible to a total valid coloring because $s_i$ can not become a paired vertex, contrary to the hypothesis.
    \item Otherwise, let $y$ be an uncolored neighbor of $s_i$. Clearly, $y \not \in D$. Since $y$ is uncolored then it has exactly one neighbor in $D'$. That is,  $s_i$ is the unique neighbor of $y$ in $D'$. Since $D\subseteq D'$ and $s_i\in D' \setminus D$, it follows that $D$ is not a dominating set,
    contradiction.
\end{itemize}
On the other hand, $C'$ and $C''$ are extensions of $C$. Then the vertices of D have the same color in these colorings. Any colored vertex $v\not \in D$ of $C'$ was obtained by some propagation rule base on previous colored vertices. The rules are correct and deterministic. Hence, $v$ must have the same color in $C''$ and $C''$ is an extension of $C'$.
\end{proof}

Clearly, given a partial valid coloring $C$, we can compute
efficiently  a stable extension $C'$ of it. In addition, if $D$ is
a dominating set then we can try to obtain a total valid coloring
from the stable coloring $C'$ by appropriately choosing exactly
one  vertex from each subset $N_U(s_i)$, to be black, that
is, to be the pair of the so far single vertex $s_i$.

\begin{lemma}
Let $U$ and $S$, respectively  be the sets of uncolored and single
vertices, relative to some stable coloring $C'$ of graph $G$. If
$C'$ extends to a total valid coloring then, for each $s_i \in S$,
$G[N_U(s_i)]$ is a union of a star and an independent set, any
of them possibly empty. Moreover, the pair of $s_i$ must be a maximum degree vertex in $G[N_U(s_i)]$.
\end{lemma}
\begin{proof}
   Suppose by contrary that $G[N_U(s_i)]$ is not a union of a star and an independent set. Then $G[N_U(s_i)]$ contains either two non-adjacent edges, or a $K_3$.
   \begin{itemize}
   		\item Let $\{(u_1,u_1), (v_1,v_2)\}$ be two disjoint edges in $G[N_U(s_i)]$. Since no white vertices can be adjacent, let $u'$ be the black vertex from $\{ u_1,u_2 \} $ and $v$ the black vertex from $\{v_1,v_2 \} $. Then $\{u,s_i,v\}$ is a black $P_3$ or $K_3$ and therefore can not be extended to a valid coloring.
   		\item Let $\{(u_1,u_2,u_3)\}$ be a $K_3$ in $G[N_U(s_i)]$. Therefore $\{s_i,u_1,u_2,u_3\}$ is a $K_4$ and therefore $G$ has no valid coloring.
   \end{itemize}
Consequently, $G[N_U(s_i)]$ must be a union of a star and an independent set.
Now, suppose by contrary that the pair of $s_i$ is a vertex $v\in N_U(s_i)$ and $v$ has not maximum degree in $G[N_U(s_i)]$. Clearly, el rest of vertices in $N_U(s_i)$ are white vertices. In particular, a maximum degree vertex $u$ in $G[N_U(s_i)]$ is white.
But, there is some neighbor $z\ne v$ of $u$ in $N_U(s_i)$ and $z$ is not adjacent to $v$. Hence, $z$ and $u$ are white adjacent vertices, which is a contradiction.
\end{proof}

We can repeatedly execute the procedure below described
for choosing the vertices to be paired to the
single vertices $s_i$ of the partial colorings. The procedure is
repeated until all parts of the partition $U = N_U(s_1)
\cup \ldots \cup N_U(s_p)$ have selected their paired black
vertices or the coloring becomes invalid.
\\
\\
Let $s_i \in S$ be a single vertex. {\it Case 1}:  $N_U(s_i)=
\emptyset$: then stop, it will not lead to a valid one. {\it Case
2}: There is exactly one maximum degree vertex in $G[N_U(s_i)]$:
then clearly, the only alternative is to choose this vertex. {\it
Case 3}: There is no edge $vw$, where $ v \in N_U(s_i)$ and $w \in
N_U(s_j)$, for any $j \ne i$: then the choice of the neighbor of
$s_i$ to become black is independent on the choices of the others
parts of the partition. Choose the vertex $w$ of maximum degree in
$G[N_U(s_i)]$ that minimizes the weight of the edge $ws_i$. {\it
Case 4}: There is an edge $vw$, where $v \in N_U(s_i)$ and $w \in
N_U(s_j)$, for some  $j \ne i$: then $v$ may become white if and
only if $w$ may become black. Each of these two choices may lead
to valid or invalid total colorings. So, we proceed with both
alternatives, as if in parallel.
\\
\\
After applying any the above Cases  2, 3 or 4,  perform the
propagation rules again and validate the coloring so far obtained.
Proceed so until eventually the coloring becomes invalid,  or a
valid solution is obtained. At the end, choose the minimum weight
solution obtained throughout the process.
\\

As for the complexity, it is clear that it depends on the
cardinality of the dominating set $D$ and on the number of
parallel iterations, considered sequentially. Next, we describe
bounds for these parameters.

\begin{lemma}
There are at most $2^{q}$ parallel computations where $q \leq p =|S|
\leq |D|$, and $q \leq \frac{n}{3}$.
\end{lemma}
\begin{proof}: By Theorem 1, it follows that $p \leq |D|$. On the other
hand, we can apply the above Cases 1-4, in such an ordering
that we keep applying Cases 1 and 2, with propagation until all
remaining single vertices $s_i$ satisfy $|N(s_i) \cap U| \geq 2$.
Let $S' \subset S$ denote the set of remaining single vertices,
and $q = |S'|$. Consequently, $q \leq \frac{n}{3}$.

Next, examine the parallel computations. They are generated by
Case 4. Let $vw$ be an edge of $G$, where $v \in N(s_i) \cap U$
and $w \in N(s_j) \cap U$, $i \neq j$. In one of the instances,
$v$ is black, meaning that $s_i$ becomes paired, while in the
other one $w$ is black, implying that $s_j$ becomes paired. This
means that the size of the set $S'$ of single vertices always
decreases by at least one unit in all computations. Hence there
are at most $2^{q}$ parallel computations.
\end{proof}

Considering that the remaining operations involved in each parallel thread of the
algorithm can be performed in linear time, it is not hard to
conclude  that there there is an $O(2^{q}m)$ time
algorithm to obtain a minimum DIM, if any, extensible from a
partial valid coloring $C$ of a weighted graph $G=(V,E)$ such that
$D=C^{-1}(black)\cup C^{-1}(white)$ is a dominating set of $G$.

The complexity of the algorithm depends on the size of the
dominating set $D$ employed. We remark that if  $G=(V,E)$ has no
isolated vertices then we can easily find in linear time a
dominating set with at most half the vertices. Just determine a
maximal independent set $I$. Clearly, $I$ and $V \setminus I$ are
both dominating sets of $G$ and one of them has at most
$\frac{n}{2}$ vertices.



Finally, in order to obtain the minimum weighted DIM of the graph
$G$, we have to apply the described algorithm for all possible
bi-colorings of $D$. There are exactly $2^{|D|}$ such colorings.
Therefore

\begin{theorem} There is an algorithm of complexity $O(min \{2^{2|D|}, 2^{\frac{5n}{6}}\}\cdot m)
\approx O^*(min \{ 4^{|D|}, 1.7818^{n}\})$ to compute a minimum
weighted DIM of a weighted graph, if existing.
\end{theorem}
\begin{proof}
The complexity is $O(2^{|D|}\cdot 2^q \cdot m)=O(2^{|D|}\cdot 2^{min\{|D|,\frac{n}{3}\}} \cdot m)=O(2^{|D|}\cdot min\{2^{|D|},2^{\frac{n}{3}}\} \cdot m)=O(min\{2^{2|D|},2^{|D|+\frac{n}{3}}\} \cdot m)=O(min\{2^{2|D|},2^{\frac{n}{2}+\frac{n}{3}}\} \cdot m)=O(min\{2^{2|D|},2^{\frac{5n}{6}}\} \cdot m)$.
\end{proof}



\begin{corollary}
The above algorithm  solves the minimum weighted DIM problem in
$O(m)$ time given a dominating set of fixed size.
\end{corollary}

\section{An algorithm based on maximal independent sets}

In this section, we describe an exact algorithm for finding a
minimal weighted DIM of a graph, based on enumerating maximal
independent sets. We consider a weighted graph $G=(V,E)$.


Any maximal independent set $I \subseteq V$ induces a partial
bi-coloring in $G$ as follows: 
\begin{itemize}
\item color as black all vertices of $V
\setminus I$ 
\item color as white the vertices of $I$ except those
having exactly one single neighbor.
\end{itemize}

\begin{observation}
If all vertices of $G$ have degree $\ne 1$ then the above partial coloring is total.
\end{observation}

The algorithm is then based on the following lemma.

\begin{lemma}
Let $G$ be a graph, $I$ a maximal independent set of it and $C$
the partial bi-coloring induced by $I$. Then $C$  is extensible to
a DIM if and only if $C$ is a valid coloring and each single
vertex, if existing, has at least one uncolored neighbor in $C$.
\end{lemma}
\begin{proof}
   $\Rightarrow)$ It is easy to see that if $C$ is not a valid coloring, then it is not extensible to a $DIM$. Besides, if $C$ has a single vertex $v$ with no uncolored neighbors then all neighbors of $v$ are white in $C$ and in any extension of $C$. Also,$C$ is not extensible to a $DIM$ because $v$ can not ever get its pair. \\
   $\Leftarrow)$ Let $C$ be a valid coloring where each single black vertex has at least one uncolored neighbor. Then for each single black vertex $v$, choose any uncolored neighbor $w$ to be its pair ($w$ has exactly one single neighbor) and the remaining of uncolored vertices get color white. In this total coloring, the black vertices induce an 1-regular subgraph and the white vertex set is an independent set because it is part of $I$. Hence, the total coloring is valid and hence a DIM.
\end{proof}

The algorithm can then be formulated as follows. Generate the
maximal independent sets $I$ of $G$. For each $I$, find its
induced coloring $C$. If $C$ is invalid or some single vertex has
no uncolored neighbor then do nothing. Otherwise, for each single
vertex $v$ in $C$, if any, choose the minimum weight $vw$, among
the uncolored neighbors of $v$; then color $w$ as black and the
remaining neighbors of $v$ as white. The set of black vertices
then forms a DIM of $G$. At the end select the minimum weight
among all DIMs obtained in the process, if any.

Clearly, this algorithm determines the minimum weight DIM of a
weighted graph $G=(V,E)$ because given any DIM $E' \subseteq E$ of
$G$, the vertex set formed by those vertices not incident to any
of the edges of $E'$ is an independent set and as such, is clearly
a subset of some maximal independent set of $G$. So, any DIM $E'$
is considered in the algorithm.


All the operations performed by the algorithm relative to a fixed
maximal independent set can be performed in linear time $O(m)$. If
$G$ has $\mu$ maximal independent sets, we can generate them all
in time  $O(nm\mu)$ time \cite{Ts}. Therefore the complexity of
the entire algorithm is $O(nm^2\mu)$. On the other hand, $\mu \leq
O(3^{\frac{n}{3}})$,
leading to a worst case of
$O(3^{\frac{n}{3}}nm^2)\approx O^*(1.44225^{n})$ time. In
particular, if $G$ is a bipartite graph then $\mu \leq
2^{\frac{n}{2}}$ and the complexity reduces to $O^*(1.41421^{n})$.
In any case, if $G$ has a polynomial number of maximal independent
sets then the algorithm terminates within polynomial time.

Finally, we observe the following additional relation between
maximal independent sets and DIM's.

\begin{lemma}
Let $G(V,E)$ be a graph with no isolated edges, $E' \subseteq E$ a DIM of $G$, and $I
\subseteq V$ the independent set formed by those vertices not
incident to any of the edges of $E'$. Then $I$ is contained in
exactly one maximal independent set of $G$.
\end{lemma}
\begin{proof}:  If $I$ is a maximal independent set there is nothing to prove. Otherwise, suppose the lemma is false and let $I_1$, $I_2$ be two distinct maximal independent sets properly containing $I$. Let $V_1 = I_1 \setminus I$, and $V_2 = I_2 \setminus I$. Choose any $v_2 \in V_2$. Clearly, $\{v_2\} \cup I$ is an indepedent set, and we know that $I_1 = V_1 \cup I$ is a maximal one. Consequently, there must be some vertex  $v_1 \in V_1$ adjacent to $v_2$. However, both $v_1$ and $v_2$ are vertices incident to edges of the DIM $E'$. Consequently, $v_1v_2 \in E'$.  In this case, $v_1v_2$ must form an isolated edge of $G$, a contradiction. Therefore the lemma holds. 
\end{proof}

Based on the above lemma and that fact that every isolated edge must be part of any DIM, it is simple to extend the exact algorithm proposed in this section, so as to count the number of
distinct DIM's (unweighted or minimum weighted) of $G$, in the same complexity as deciding whether $G$ admits a DIM. Observe that $G$ may contain an exponential number of DIM's.

\end{document}